\documentclass[
  aps,
  pra,
  reprint,
  superscriptaddress,
  nofootinbib,
  longbibliography,
  floatfix
]{revtex4-2}

\usepackage{amsmath,amssymb,mathtools,amsthm}
\usepackage{bm}
\usepackage{microtype}
\usepackage{tikz}
\usepackage{hyperref}
\hypersetup{hidelinks}

\theoremstyle{plain}
\newtheorem{theorem}{Theorem}

\theoremstyle{definition}

\theoremstyle{remark}

\newcommand{\F}{\mathbb{F}}
\newcommand{\R}{\mathbb{R}}

\newcommand{\Tr}{\operatorname{Tr}}
\newcommand{\rank}{\operatorname{rank}}

\newcommand{\Span}{\operatorname{span}}

\newcommand{\FKempty}{\mathcal{F}_{\overline{K_4}}}
\newcommand{\FCfour}{\mathcal{F}_{C_4}}

\newcommand{\AnticommutationGraphs}{%
\begin{tikzpicture}[
  x=1cm,y=1cm,line width=0.65pt,
  vertex/.style={circle,fill=black,inner sep=0pt,minimum size=3.4pt},
  every node/.style={font=\small},
  every label/.style={font=\small,inner sep=1pt}]
  \useasboundingbox (-0.40,-0.85) rectangle (7.70,2.15);
  \begin{scope}[xshift=0.15cm]
    \node (titlea) at (0.65,1.94) {\phantom{$K_3\sqcup K_1$}};
    \node at (0.65,1.94) {$\overline{K_4}$};
    \node[anchor=east,xshift=-4pt] at (titlea.west) {(a)};
    \node[vertex,label=above:$u_1$] at (0,1.15) {};
    \node[vertex,label=above:$u_2$] at (1.3,1.15) {};
    \node[vertex,label=below:$u_3$] at (1.3,0) {};
    \node[vertex,label=below:$u_4$] at (0,0) {};
    \node at (0.65,-0.72) {$|E|=0$};
  \end{scope}
  \begin{scope}[xshift=3.0cm]
    \node (titleb) at (0.65,1.94) {\phantom{$K_3\sqcup K_1$}};
    \node at (0.65,1.94) {$C_4$};
    \node[anchor=east,xshift=-4pt] at (titleb.west) {(b)};
    \draw (0,1.15) -- (1.3,1.15) -- (1.3,0) -- (0,0) -- cycle;
    \node[vertex,label=above:$u_1$] at (0,1.15) {};
    \node[vertex,label=above:$u_2$] at (1.3,1.15) {};
    \node[vertex,label=below:$u_3$] at (1.3,0) {};
    \node[vertex,label=below:$u_4$] at (0,0) {};
    \node at (0.65,-0.72) {$|E|=4$};
  \end{scope}
  \begin{scope}[xshift=5.85cm]
    \node (titlec) at (0.65,1.94) {$K_3\sqcup K_1$};
    \node[anchor=east,xshift=-4pt] at (titlec.west) {(c)};
    \draw (0,1.15) -- (1.3,1.15) -- (1.3,0) -- cycle;
    \node[vertex,label=above:$u_1$] at (0,1.15) {};
    \node[vertex,label=above:$u_2$] at (1.3,1.15) {};
    \node[vertex,label=below:$u_3$] at (1.3,0) {};
    \node[vertex,label=below:$u_4$] at (0,0) {};
    \node at (0.65,-0.72) {$|E|=3$};
  \end{scope}
\end{tikzpicture}%
}

\begin{document}

\title{Quartic Certificates for Pure-State Tomography with Pauli Measurements}
\author{Hengjian Li}
\affiliation{Department of Physics, The University of Texas at Dallas, Richardson, Texas 75080, USA}
\author{Baisong Sun}
\affiliation{Department of Physics, The University of Texas at Dallas, Richardson, Texas 75080, USA}
\author{Bei Zeng}
\affiliation{Department of Physics, The University of Texas at Dallas, Richardson, Texas 75080, USA}

\begin{abstract}
Full Pauli-basis tomography of an arbitrary $n$-qubit state uses the expectation values of all nonidentity Pauli observables. For multiqubit systems, prior knowledge that the state is pure allows some of these observables to be omitted while retaining uniqueness. We study two resulting notions of uniqueness: whether every pure state is uniquely determined among pure states (UDP), and whether every pure state is uniquely determined among all states, including mixed states (UDA). We derive quartic uniqueness certificates from the fourth elementary symmetric polynomial of operators in the real span of the omitted Pauli observables. An exact quartic identity shows that, for $n$-qubit systems with $n\geq 3$, every pure state is UDA whenever the omitted Pauli observables contain no four distinct operators that commute pairwise and multiply to $\pm I$. For three qubits, this condition is also necessary for UDP, yielding an exact characterization of both UDP and UDA. We further prove that UDP and UDA coincide for any number of qubits whenever the omitted Pauli observables commute pairwise. These results provide directly checkable structural criteria for certifying pure-state tomography from restricted Pauli measurements.
\end{abstract}

\maketitle

\section{Introduction}
\label{sec:intro}

Quantum-state tomography connects the preparation of a quantum system to
an experimentally testable description of its state \cite{Prugovecki1977}. Reliable state
characterization is essential for assessing quantum devices
\cite{ParisRehacek2004,EisertEtAl2020,SosaMartinezEtAl2017}, but determining a general state of
$n$ qubits requires information about $4^n-1$ independent real parameters.
Prior knowledge that the state is pure can substantially reduce the
measurement requirements \cite{Weigert1992,AmietWeigert1999,FlammiaSilberfarbCaves2005,HeinosaariMazzarellaWolf2013,KechVranaWolf2015,GoyenecheEtAl2015,KechWolf2017}. Pauli
observables are particularly useful because their expectation values can
be obtained from local spin measurements. This raises a basic design
question: which Pauli observables can remain unmeasured while the
retained data still identify every pure state? We call the unmeasured
Pauli observables \emph{omitted} and the measured ones \emph{retained}.

The meaning of ``identify'' matters. A pure state is \emph{uniquely
determined among pure states} (UDP) if no other pure state produces the
same data, and \emph{uniquely determined among all states} (UDA) if no
other density operator, pure or mixed, does so \cite{ChenEtAl2013}.
UDA therefore excludes a mixed-state explanation of the measured data.
It also permits reconstruction over the convex set of density operators
without imposing a purity constraint: for suitable measurements, positivity
and the data already single out the state
\cite{KalevKosutDeutsch2015,BaldwinDeutschKalev2016,CarmeliEtAl2016}.
We require these uniqueness properties to hold for \emph{every} pure
state. This universal requirement is stronger than identifying generic
pure states while allowing exceptional ambiguities
\cite{Finkelstein2004,CarmeliEtAl2014,CarmeliEtAl2015}.

Pauli tomography has been studied from both reconstruction and
measurement-design perspectives. Random Pauli observables support
compressed-sensing recovery of low-rank states \cite{GrossEtAl2010,FlammiaEtAl2012},
and recent work gives approximate pure-state reconstruction with
near-optimal copy and time complexity using nonadaptive Pauli product-basis
measurements \cite{GrewalEtAl2026}. 
For exact universal UDA, Ma et al.\ established the optimal
counts of $11$ Pauli observables for two qubits and $31$ for three
qubits, including the identity in both cases \cite{MaEtAl2016}.
The three-qubit analysis used failing
sets studied in Jackson's thesis \cite{Jackson2013} and a detailed
spectral calculation for a selected measurement scheme. Variational
searches subsequently found further schemes and numerical agreement
between UDP and UDA \cite{ZhangZhuZeng2024}; optimization also provides
tests of state-dependent uniqueness \cite{WuEtAl2025}. These developments
motivate an analytic question: which structural properties of an
arbitrary set of omitted observables guarantee universal uniqueness?

Our approach uses a spectral quantity to certify this guarantee. Two
states give identical retained expectation values precisely when their
difference lies in the real span of the omitted Pauli observables.
Every operator in this span is Hermitian and traceless.
For such an operator $A$, with eigenvalues
$\lambda_1,\ldots,\lambda_d$ and $d=2^n$, define
\begin{equation}
  e_4(A)=\sum_{1\leq i<j<k<\ell\leq d}
  \lambda_i\lambda_j\lambda_k\lambda_\ell.
  \label{eq:e4-spectrum}
\end{equation}
This is the fourth elementary symmetric polynomial: the subscript $4$
counts the four eigenvalue factors in each term. As a function of the
Pauli coefficients of $A$, it is a homogeneous polynomial of degree
four, which explains ``quartic'' in the title. Its positivity provides
the certificate. A nonzero traceless Hermitian operator with $e_4(A)>0$
has at least two positive and two negative eigenvalues. It therefore
cannot be the difference between a pure state and a distinct density
operator, since such a difference has exactly one positive eigenvalue.
Proving $e_4(A)>0$ for every nonzero operator in the real span of
the omitted Pauli observables thus certifies universal UDA.

We derive an exact identity for $e_4$ that yields this positivity from
the structure of the omitted observables. Call their set
\emph{four-set-free} when it contains no four distinct observables
that commute pairwise and multiply to $\pm I$. 
For every $n\geq3$, four-set-freeness guarantees $e_4(A)>0$
for every nonzero operator $A$ in the real span of the omitted
Pauli observables, and hence universal UDA.
This gives a finite test using only Pauli
multiplication and commutation. For three qubits, the condition is
also necessary: each such four-set hides a population difference
between two orthogonal pure states. We thereby obtain an exact
characterization of both UDP and UDA for every three-qubit Pauli
measurement scheme and an analytic explanation of the $945$ minimal
failing sets identified in earlier work \cite{Jackson2013,MaEtAl2016}.
For $n\geq4$, four-set-freeness remains sufficient but is no longer
necessary; eigenvalue multiplicities can preserve uniqueness even
when a four-set is omitted.

We also prove that universal UDP and universal UDA are equivalent
whenever the omitted Pauli observables commute pairwise, for any
number of qubits. To analyze this class, we apply a Clifford
transformation that makes all omitted observables diagonal in the
computational basis. In these transformed coordinates, every
nondiagonal Pauli observable is retained, so states with identical
retained data can differ only in their diagonal entries. We show
that a failure of UDA then forces two distinct computational-basis
states to have identical retained data, and therefore also implies
a failure of UDP. For this commuting class, we obtain an exact
criterion: after the Clifford transformation, universal UDP and
UDA hold if and only if the retained diagonal Pauli observables
distinguish every pair of computational-basis states. Together
with the quartic certificate, this result shows that any separation
between universal UDP and UDA for $n\geq4$ requires a noncommuting
omitted set containing four distinct observables that commute
pairwise and multiply to $\pm I$.

The remainder of the paper is organized as follows.
Section~\ref{sec:prelim} reviews the measurement-kernel criteria and
the binary representation of Pauli observables, including the
three-qubit failing-set example. Section~\ref{sec:quartic} derives the
quartic identity and its positive semidefinite decomposition.
Section~\ref{sec:tomography} applies this identity to four-set-free
supports, the exact three-qubit characterization, and the higher-qubit
case. Section~\ref{sec:general} proves UDP--UDA equivalence for
commuting omitted supports and analyzes the associated eigenvalue
multiplicities. Finally, Sec.~\ref{sec:discussion} summarizes the
results and discusses open questions.

\section{Preliminaries}
\label{sec:prelim}

To turn the measurement-design question into a structural problem, we
first describe the changes of a state that are invisible to the retained
measurements. We express these changes through the measurement kernel,
translate pure-state uniqueness into spectral conditions, and then encode
the omitted Pauli observables by binary labels. A three-qubit failing set
illustrates how an omitted commuting four-set creates an ambiguity.

\subsection{Measurement data and indistinguishable states}
\label{sec:measurement-kernel}

A state on a $d$-dimensional Hilbert space is represented by a density
operator $\rho\succeq0$ with $\Tr\rho=1$, and it is pure when
$\rho=|\psi\rangle\langle\psi|$ for a normalized vector $|\psi\rangle$.
Given a linearly independent family of Hermitian observables
$\mathcal A=\{I,A_1,\ldots,A_m\}$, define its linear measurement map
on an arbitrary Hermitian operator $\Delta$ by
\begin{equation}
  \mathcal{M}_{\mathcal A}(\Delta)
  =\bigl(\Tr \Delta,\Tr(A_1\Delta),\ldots,\Tr(A_m\Delta)\bigr).
\end{equation}
For a state $\rho$, the vector $\mathcal M_{\mathcal A}(\rho)$ is the
measurement data: its components $\Tr(A_j\rho)$ are the measured
expectation values, while $\Tr\rho=1$ records normalization. Two states
$\rho$ and $\sigma$ produce the same data precisely when their difference
belongs to the Hermitian kernel
\begin{equation}
  \mathcal{K}_{\mathcal A}
  =\{\Delta=\Delta^\dagger:\mathcal{M}_{\mathcal A}(\Delta)=0\}.
\end{equation}
This kernel consists of Hermitian changes that are invisible to the
measured observables, and every such change is traceless.

Given the observable family $\mathcal A$, a pure state is uniquely
determined among pure states (UDP) if no other pure state gives the same
measurement data, and uniquely determined among all states (UDA) if no
other density operator does so. We say that $\mathcal A$ has universal
UDP when every pure state is UDP, and universal UDA when every pure state
is UDA. The kernel $\mathcal K_{\mathcal A}$ converts these global
uniqueness requirements into spectral conditions. For any nonzero
$\Delta\in\mathcal K_{\mathcal A}$, let $\Delta_+$ and $\Delta_-$ denote
its positive and negative spectral parts, respectively. They are positive
semidefinite with orthogonal supports and satisfy
$\Delta=\Delta_+-\Delta_-$. Tracelessness gives
$\Tr \Delta_+=\Tr \Delta_-=t>0$, and hence
\begin{equation}
  \Delta=t(\rho_+-\rho_-),\qquad \rho_\pm=\Delta_\pm/t.
  \label{eq:kernel-state-difference}
\end{equation}
The density operators $\rho_+$ and $\rho_-$ give identical measurement
data. The state $\rho_+$ is pure exactly when $\Delta$ has one positive
eigenvalue, and $\rho_-$ is pure exactly when $\Delta$ has one negative
eigenvalue. In particular, both are pure when $\Delta$ has rank two.
For the reverse direction, suppose that a pure state
$\rho_\phi=|\phi\rangle\langle\phi|$ and a distinct density operator
$\sigma$ give identical measurement data. Then
$\rho_\phi-\sigma\in\mathcal K_{\mathcal A}$.
For every vector $|v\rangle$ orthogonal to $|\phi\rangle$, one has
\begin{equation}
  \langle v|(\rho_\phi-\sigma)|v\rangle
  =-\langle v|\sigma|v\rangle\leq0.
\end{equation}
Hence $\rho_\phi-\sigma$ has at most one positive eigenvalue. Because
this difference is nonzero and traceless, it must also have a positive
eigenvalue and therefore has exactly one. Reversing the sign shows that
$\sigma-\rho_\phi$ has exactly one negative eigenvalue.

Let $n_+(\Delta)$ and $n_-(\Delta)$ count the strictly positive and
strictly negative eigenvalues, including multiplicities. Under the
correspondence above, a nonzero kernel operator gives a pure--pure ambiguity
exactly when it has rank two. It gives an ambiguity between a pure state
and another density operator exactly when $n_+(\Delta)=1$ or
$n_-(\Delta)=1$. Consequently, the standard kernel criteria are
\begin{equation}
  \text{UDP:}\quad \rank \Delta\geq3,
  \label{eq:udp-kernel}
\end{equation}
\begin{equation}
  \text{UDA:}\quad n_+(\Delta)\geq2,\qquad n_-(\Delta)\geq2,
  \label{eq:uda-kernel}
\end{equation}
in each case for every nonzero $\Delta\in\mathcal K_{\mathcal A}$
\cite{HeinosaariMazzarellaWolf2013,ChenEtAl2013,CarmeliEtAl2014,
BaldwinDeutschKalev2016}.
We call $(n_+(\Delta),n_-(\Delta))$ the inertia of $\Delta$. A one-sided
inertia obstruction means that one of these two counts is one. Thus the
possible gap between UDP and UDA is a pure--mixed ambiguity without any
pure--pure ambiguity in the same measurement scheme. These criteria apply
to arbitrary Hermitian observables; the Pauli structure enters through the
particular form of their kernel.

\subsection{Pauli measurements and binary labels}
\label{sec:pauli-notation}

We now specialize to Pauli measurements on $n$ qubits, with $d=2^n$.
For a single qubit, the Pauli matrices are
\begin{equation}
  X=\begin{pmatrix}0&1\\1&0\end{pmatrix},\quad
  Y=\begin{pmatrix}0&-i\\i&0\end{pmatrix},\quad
  Z=\begin{pmatrix}1&0\\0&-1\end{pmatrix}.
\end{equation}
They represent measurements along the three Bloch-sphere axes, with
outcomes $\pm1$. An $n$-qubit \emph{Pauli string} is a tensor product
of $I,X,Y,Z$, such as $XZI=X\otimes Z\otimes I$. Its expectation
value is the average product of the corresponding local measurement
outcomes; a qubit carrying an identity factor does not contribute to
this product.
Every nonidentity Pauli string is traceless and has eigenvalues $\pm1$.
Together with the identity, the $4^n$ strings form an orthogonal basis
of Hermitian operators.

Binary labels give a compact description of these same observables.
For each qubit, assign a pair of bits according to
\begin{equation}
  \begin{gathered}
  I\leftrightarrow(0,0),\qquad X\leftrightarrow(1,0),\\
  Z\leftrightarrow(0,1),\qquad Y\leftrightarrow(1,1).
  \end{gathered}
\end{equation}
Here $\F_2=\{0,1\}$ is the binary field, with arithmetic modulo two.
Collecting the first and second bits into $x,z\in\F_2^n$ gives the
label $u=(x,z)\in V_n=\F_2^{2n}$; for example, $XZI$ has label
$(100,010)$. A binary label specifies a Pauli string only up to an
overall phase. Because a measured observable must be Hermitian, we fix
this phase by defining
\begin{equation}
  \begin{aligned}
  P_u&=i^{x\cdot z}X^xZ^z,\\
  X^x&=\bigotimes_{j=1}^nX^{x_j},\qquad
  Z^z=\bigotimes_{j=1}^nZ^{z_j}.
  \end{aligned}
  \label{eq:pauli-convention}
\end{equation}
Here $x\cdot z=\sum_{j=1}^n x_jz_j$ is evaluated over the integers. The
phase convention makes the local label $(1,1)$ represent $Y=iXZ$, and
gives $P_0=I$ and $P_u^2=I$ for every $u\in V_n$. Multiplication of
Pauli strings adds their binary labels, componentwise modulo two, while
the overall phase is kept separately. In particular, a product is
proportional to the identity exactly when the sum of its labels is zero.

Commutation also has a binary description. For $v=(x',z')$, define
\begin{equation}
  \omega(u,v)=x\cdot z'+z\cdot x'\pmod2.
  \label{eq:symplectic}
\end{equation}
The form $\omega$ is bilinear, alternating, and nondegenerate; in
particular, $\omega(u,u)=0$, and only $u=0$ commutes with every label.
The binary symplectic form records commutation, while the trace relation
expresses the orthogonality of the Pauli basis:
\begin{equation}
  P_uP_v=(-1)^{\omega(u,v)}P_vP_u,
  \qquad \Tr(P_uP_v)=d\delta_{u,v}.
  \label{eq:pauli-basic}
\end{equation}
Thus $\omega(u,v)=0$ for commuting pairs and $1$ for anticommuting
pairs. If two strings commute, their product is Hermitian, so the
phase relating it to the chosen representative is a real sign:
\begin{equation}
  P_uP_v=\epsilon(u,v)P_{u+v},\qquad \epsilon(u,v)\in\{\pm1\}.
  \label{eq:epsilon}
\end{equation}

A $2^n\times2^n$ unitary operator $U$ is called a
\emph{Clifford unitary} if, for every Pauli string $P$, there is a
Pauli string $P'$ such that
\begin{equation}
  UPU^\dagger=\pm P'.
\end{equation}
The set of all such unitaries is the Clifford group. For example, the
single-qubit Hadamard gate is
\begin{equation}
  H=\frac1{\sqrt2}\begin{pmatrix}1&1\\1&-1\end{pmatrix}.
\end{equation}
Conjugation by $H$ exchanges $X$ and $Z$ and reverses the sign of $Y$,
as expressed by
\begin{equation}
  HXH^\dagger=Z,\qquad HZH^\dagger=X,\qquad HYH^\dagger=-Y.
\end{equation}
In the binary description, the Hadamard gate exchanges the labels
$(1,0)$ and $(0,1)$ of $X$ and $Z$. Its action is therefore
\begin{equation}
  \begin{pmatrix}x\\z\end{pmatrix}
  \longmapsto M_H\begin{pmatrix}x\\z\end{pmatrix},\qquad
  M_H=\begin{pmatrix}0&1\\1&0\end{pmatrix}.
\end{equation}
Thus $M_H$ fixes the label $(1,1)$ of $Y$, while the sign in
$HYH^\dagger=-Y$ is not encoded by the binary label.

For a general $n$-qubit Clifford unitary, conjugation preserves Pauli
multiplication, so its action on labels preserves binary addition.
This action is therefore represented by an invertible $2n\times2n$
binary matrix $M_U$. The operator and binary descriptions are related by
\begin{equation}
  UP_uU^\dagger=\eta_U(u)P_{M_Uu},\qquad
  \eta_U(u)\in\{\pm1\}.
\end{equation}
Here $u=(x,z)$ is treated as a column of $2n$ bits, and multiplication
by $M_U$ is over $\F_2$. Thus $U$ acts on quantum states, whereas
$M_U$ records how it permutes the Pauli labels, with the signs stored
in $\eta_U$. Conjugation preserves commutation, so
\begin{equation}
  \omega(M_Uu,M_Uv)=\omega(u,v).
\end{equation}
An invertible binary map with this property is called symplectic.
Conversely, every symplectic map of the labels is implemented by a
Clifford unitary, up to individual Pauli signs
\cite{Gottesman1997,DehaeneDeMoor2003}. Transforming the state and
measured observables by the same Clifford unitary preserves all expectation
values, while the possible Pauli signs only relabel the $\pm1$ outcomes.
Clifford transformations therefore relate Pauli measurement schemes with
the same distinguishability properties.

For a Pauli measurement scheme, let $S\subseteq V_n\setminus\{0\}$
be the measured support and
$T=(V_n\setminus\{0\})\setminus S$ its omitted complement. Pauli
orthogonality identifies the kernel with the operator span of $T$. More
generally, for any label set $R\subseteq V_n\setminus\{0\}$, define
\begin{equation}
  \mathcal V_R:=\Span_{\R}\{P_u:u\in R\},
  \qquad
  \mathcal K_{\mathcal A}=\mathcal V_T.
  \label{eq:pauli-kernel}
\end{equation}
Thus the kernel consists exactly of real combinations of the omitted
Pauli observables. Throughout, spans of labels are over $\F_2$, whereas
spans of Hermitian operators are over $\R$. Our task is to decide whether
this span can contain the state differences described in
Sec.~\ref{sec:measurement-kernel},
using only the choice of omitted observables.

\subsection{Failing sets and pure-state ambiguity}
\label{sec:failing-sets}

The connection between omitted observables and indistinguishable states
is captured by the notion of a \emph{failing set}. A failing set is a
collection of nonidentity Pauli observables whose real span contains a
nonzero operator with only one positive or only one negative eigenvalue
\cite{MaEtAl2016,Jackson2013}. Omitting every observable in such a set
therefore prevents universal UDA. A concrete three-qubit example is
$\{IIZ,IZI,ZII,ZZZ\}$, for which
\begin{equation}
  \begin{split}
  &|000\rangle\langle000|-|111\rangle\langle111|\\
  &\qquad=\frac14(IIZ+IZI+ZII+ZZZ).
  \end{split}
  \label{eq:standard-failing}
\end{equation}
The four terms on the right are exactly the nonidentity Pauli observables
whose expectations differ between $|000\rangle$ and $|111\rangle$. If
all four are omitted, the projector difference in
Eq.~\eqref{eq:standard-failing} belongs to $\mathcal K_{\mathcal A}$, so
the two pure states give identical retained measurement data. By
linearity, every mixture
\begin{equation}
  p|000\rangle\langle000|+(1-p)|111\rangle\langle111|,
  \qquad 0\leq p\leq1,
\end{equation}
gives the same retained data. The omitted observables
therefore conceal the relative population of the two orthogonal states,
producing both pure--pure and pure--mixed ambiguities. Conjugating both
the states and the omitted observables by the same Clifford unitary
produces further examples with the same ambiguity.

The example above is a special case of four distinct nonidentity Pauli
strings that commute pairwise and multiply to $\pm I$. If
$F=\{u_1,u_2,u_3,u_4\}$ denotes their binary labels, the multiplication
and commutation rules above make these conditions equivalent to
\begin{equation}
  \begin{aligned}
  u_1+u_2+u_3+u_4&=0,\\
  \omega(u_i,u_j)&=0\quad(1\leq i<j\leq4).
  \end{aligned}
  \label{eq:fourset}
\end{equation}
We call this configuration a \emph{commuting zero-sum four-set}.
For the example above, the labels have $x=000$ and respective
$z$-components $001,010,100,111$. All symplectic products vanish,
and $001+010+100+111=000$ in binary arithmetic. The binary labels
capture the configuration without retaining the sign of the Pauli
product.
A support $T$ is \emph{four-set-free} if it contains no such
configuration. At three qubits, these four-sets are precisely the
Clifford images of the example above, up to Pauli signs. The quartic
certificate in Sec.~\ref{sec:tomography} will establish that excluding
all of them is sufficient for universal UDA. Thus this particular
family of failing sets gives an exact three-qubit criterion.

The state interpretation also explains the role of dimension. Put
$x=u_1$, $p=u_1+u_2$, and $q=u_1+u_3$, so that
$F=\{x,x+p,x+q,x+p+q\}$. The labels $x,p,q$ are independent and
commuting: a nontrivial dependence would make two members of $F$
equal or one of them zero. Their common eigenvalues label eight
joint eigenspaces, each of dimension $d/8$; in particular, such a
four-set requires $n\geq3$. The operator
\begin{equation}
  P_x(I+P_p)(I+P_q)\in \mathcal V_F
  \label{eq:fourset-factor}
\end{equation}
acts only on the sector where $P_p=P_q=+1$. Within that sector,
it has eigenvalue $+4$ or $-4$ according to the eigenvalue of $P_x$.
Both signs have multiplicity $d/8$, so its rank is
\begin{equation}
  d/4=2^{n-2}.
  \label{eq:fourset-rank}
\end{equation}
For three qubits, $d/8=1$, so the positive and negative eigenspaces are
both one-dimensional. The operator is therefore proportional to the
difference of two pure-state projectors. For $n\geq4$, $d/8\geq2$, and
normalizing its positive and negative spectral parts gives two mixed
states, each of rank $d/8$. This particular kernel operator therefore
witnesses only a mixed--mixed ambiguity and does not by itself violate
pure-state uniqueness. Section~\ref{subsec:commuting-ranks} extends the
calculation to the entire four-operator span, showing that every nonzero
operator in this span has equally many positive and negative eigenvalues,
with at least $d/8$ of each sign.

One can check for these configurations directly from the omitted
observables. Over $\F_2$, four labels sum to zero exactly when the two
pairs in any partition have the same sum. Group the unordered commuting
pairs in $T$ by this common binary sum $w$. Two distinct pairs
$\{u,u+w\}$ and $\{r,r+w\}$ in
the same group form a commuting zero-sum four-set exactly when
\begin{equation}
  \omega(u,r)=0.
  \label{eq:bucket-test}
\end{equation}
The pairs are disjoint and their four labels sum to zero. Since
$\omega(u,w)=\omega(r,w)=0$, every cross-pair symplectic product
equals $\omega(u,r)$. Every four-set is detected by its three
partitions into pairs. This test uses only Pauli commutation and
binary addition, without searching over states or continuous
operator coefficients.

\section{A general quartic identity for Pauli operators}
\label{sec:quartic}

To obtain a criterion that can be checked directly from the omitted
Pauli observables, we convert the spectral uniqueness condition into an
algebraic identity. The derivation has three steps. We first express the
fourth elementary symmetric polynomial through traces and identify the
relevant eigenvalue patterns. We next expand it in commuting Pauli pairs
and classify the resulting zero-sum four-label terms by their
anticommutation graphs. Finally, we absorb the noncommuting four-label
contributions into positive-semidefinite quadratic forms, leaving only the
pairwise-commuting four-sets explicit. The resulting identity is applied in
Sec.~\ref{sec:tomography}.

\subsection{A quartic criterion for pure-state uniqueness}
\label{subsec:quartic-criterion}

For a tuple $x=(x_1,\ldots,x_q)$, the elementary symmetric polynomial of
degree $j$ is
\begin{equation}
 e_j(x)=\sum_{1\leq i_1<\cdots<i_j\leq q}x_{i_1}\cdots x_{i_j}.
  \label{eq:elementary-symmetric-definition}
\end{equation}
We set $e_0(x)=1$ and $e_j(x)=0$ for $j>q$.
For a Hermitian operator $A$ with eigenvalues
$\lambda_1,\ldots,\lambda_d$, counted with multiplicity, we evaluate
the same polynomial on its spectrum and write
\begin{equation}
  e_j(A):=e_j(\lambda_1,\ldots,\lambda_d).
\end{equation}
Taking $j=4$ recovers the spectral definition of $e_4(A)$ introduced
in Eq.~\eqref{eq:e4-spectrum}. We use this fourth-degree polynomial
to exclude the eigenvalue patterns responsible for pure-state ambiguity.

Two properties make $e_4$ useful for this purpose. First, it can be
calculated from traces without diagonalizing $A$. Every measurement
kernel operator is traceless, and for $\Tr A=0$ we have
\begin{equation}
  e_4(A)=\frac{[\Tr(A^2)]^2}{8}-\frac{\Tr(A^4)}{4}.
  \label{eq:e4-trace}
\end{equation}
To obtain this expression, let $p_j=\Tr(A^j)$ denote the power sums
of the eigenvalues. Newton's identities give $e_2=-p_2/2$ and
$4e_4=-e_2p_2-p_4$ when $p_1=e_1=0$.

Second, positivity of $e_4$ guarantees the eigenvalue signs needed
for UDA. For any nonzero traceless Hermitian $A$,
\begin{equation}
  e_4(A)>0
  \quad\Longrightarrow\quad
  n_+(A),n_-(A)\geq2.
  \label{eq:e4-inertia}
\end{equation}
The reason is that a single positive eigenvalue forces $e_4$ to be
nonpositive. Write such a nonzero spectrum as
$\lambda,-x_1,\ldots,-x_m$, with $\lambda,x_i>0$, and let
$x=(x_1,\ldots,x_m)$. Tracelessness fixes $\lambda=e_1(x)$. If
$m\leq2$, the rank is at most three, so every product
of four eigenvalues contains a zero factor and $e_4(A)=0$.
If $m\geq3$, we separate the terms in $e_4(A)$ into those that contain
$\lambda$ and those that do not. This gives
\begin{equation}
  e_4(A)=e_4(x)-\lambda e_3(x)
        =e_4(x)-e_1(x)e_3(x)<0.
\end{equation}
Indeed, $e_1(x)e_3(x)$ contains four copies of every term in $e_4(x)$,
together with strictly positive terms having a repeated index.
The case of a single negative eigenvalue follows by replacing $A$
with $-A$, which leaves $e_4$ unchanged. These are precisely the
one-sided inertia obstructions described in
Sec.~\ref{sec:measurement-kernel}.
For a Pauli combination, Eq.~\eqref{eq:e4-trace} therefore expresses
$e_4(A)$ as a homogeneous quartic polynomial in its real Pauli
coefficients.

\subsection{Pauli expansion and anticommutation graphs}
\label{subsec:graph-expansion}

With real Pauli coefficients $a_u$, write
\begin{equation}
  A=\sum_{u\neq0}a_uP_u,\qquad
  s=\sum_{u\neq0}a_u^2.
\end{equation}
Since every Pauli string squares to the
identity, the diagonal terms of $A^2$ give $sI$. Anticommuting cross
terms cancel, whereas a commuting pair with $u+v=w$ contributes to the
coefficient $c_w$. For each $w\neq0$, define
\begin{equation}
  c_w=2\!\sum_{\substack{u<v\\u+v=w,\ \omega(u,v)=0}}
  \epsilon(u,v)a_ua_v.
  \label{eq:cw}
\end{equation}
Here and below, all labels are nonzero, $u<v$ refers to a fixed
ordering, and for commuting pairs the sign $\epsilon(u,v)$ is defined by
$P_uP_v=\epsilon(u,v)P_{u+v}$. Thus
\begin{equation}
  A^2=sI+\sum_{w\neq0}c_wP_w.
  \label{eq:A2}
\end{equation}
Pauli orthogonality now gives both traces directly:
\begin{equation}
  \Tr(A^2)=ds,\qquad
  \Tr(A^4)=d\left(s^2+\sum_{w\neq0}c_w^2\right).
  \label{eq:pauli-traces}
\end{equation}
The second expression follows by taking the trace of $(A^2)^2$, since
products of distinct Pauli basis elements have zero trace.
Substituting into Eq.~\eqref{eq:e4-trace} gives
\begin{equation}
  e_4(A)=\frac{d(d-2)}8s^2-\frac d4\sum_{w\neq0}c_w^2.
  \label{eq:e4-cw}
\end{equation}

Expanding the square of Eq.~\eqref{eq:cw}, we obtain
\begin{equation}
\begin{aligned}
 c_w^2={}&4\!\sum_{\substack{u<v\\u+v=w,\ \omega(u,v)=0}}a_u^2a_v^2
 \\
 &+8\!\sum_{\substack{u<v,\ r<t,\ (u,v)<(r,t)\\
 u+v=r+t=w\\\omega(u,v)=\omega(r,t)=0}}
 \epsilon(u,v)\epsilon(r,t)a_ua_va_ra_t.
\end{aligned}
 \label{eq:cw-square}
\end{equation}
Here the pairs are ordered lexicographically, so each pair of distinct
pairs occurs once in the second sum.
The cross terms in Eq.~\eqref{eq:cw-square} involve two distinct
pairs with the same sum. Such pairs cannot share a label, since a
shared label and a shared sum would make the pairs identical. Their
union is therefore a four-element set $F=\{u_1,u_2,u_3,u_4\}$ with
binary sum zero. Its three partitions into two unordered pairs are
\begin{equation}
  \begin{gathered}
    \{\{u_1,u_2\},\{u_3,u_4\}\},\\
    \{\{u_1,u_3\},\{u_2,u_4\}\},\\
    \{\{u_1,u_4\},\{u_2,u_3\}\}.
  \end{gathered}
  \label{eq:pair-partitions}
\end{equation}
Each of these partitions satisfies the equal-sum condition in
Eq.~\eqref{eq:cw-square} and contributes exactly when both pairs commute.

To organize these contributions, associate an anticommutation graph
with $F$: its vertices are the four labels, and an edge joins $u,v$
when $\omega(u,v)=1$. Such graphs also constrain jointly attainable
Pauli expectation values through bounds on the joint numerical range
\cite{XuSchwonnekWinter2024}. Here we use them to organize the terms
of the quartic expansion. Write $K_r$ for the complete graph on $r$
vertices, $\overline{K_r}$ for the edgeless graph, and $C_r$ for the
cycle on $r\geq3$ vertices; $\sqcup$ denotes disjoint union.
The zero-sum condition strongly restricts this graph. For each $i$,
bilinearity gives
\begin{equation}
 \sum_{j\neq i}\omega(u_i,u_j)
 =\omega\!\left(u_i,\sum_{j\neq i}u_j\right)
 =\omega(u_i,u_i)=0
 \label{eq:even-degree}
\end{equation}
in $\F_2$. Thus every vertex has degree $0$ or $2$. Each nonisolated
component is a cycle, leaving only the three graphs
$\overline{K_4}$, $C_4$, and $K_3\sqcup K_1$ shown in
Fig.~\ref{fig:anticommutation-graphs}.

\begin{figure}
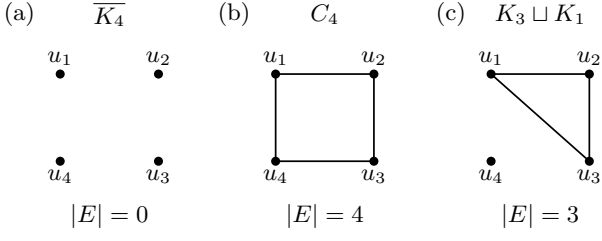

  \centering
  \AnticommutationGraphs
  \caption{Anticommutation graphs for four distinct nonzero Pauli labels whose binary sum is zero. Each vertex represents a label, each edge joins a pair of anticommuting labels, and $|E|$ gives the number of edges. The three possibilities are (a) the edgeless graph $\overline{K_4}$, in which all three partitions into two pairs are commuting; (b) the four-cycle $C_4$, in which only the partition along its two nonedges is commuting; and (c) a triangle with an isolated vertex, $K_3\sqcup K_1$, in which no partition consists of two commuting pairs.}
  \label{fig:anticommutation-graphs}
\end{figure}

For each graph $G$, let $\mathcal F_G$ denote the family of four-element
subsets of $V_n\setminus\{0\}$ with binary sum zero and anticommutation
graph isomorphic to $G$. When $F\in\FKempty$, all four Pauli strings
commute, so their product defines an order-independent sign
$\tau_F\in\{+1,-1\}$ through $\prod_{u\in F}P_u=\tau_FI$.
All three partitions in Eq.~\eqref{eq:pair-partitions} are commuting,
and each has $\epsilon(u,v)\epsilon(r,t)=\tau_F$.
Each partition contributes $-2d\,\tau_F\prod_{x\in F}a_x$ to
Eq.~\eqref{eq:e4-cw}, for a total of
$-6d\,\tau_F\prod_{x\in F}a_x$.

For $F\in\FCfour$, only the partition along the two nonedges consists
of commuting pairs. Writing it as $\{\{u,v\},\{r,t\}\}$, define
\begin{equation}
 \sigma_F=\epsilon(u,v)\epsilon(r,t)\in\{+1,-1\},
 \qquad F\in\FCfour.
 \label{eq:sigmaF}
\end{equation}
This partition contributes $-2d\,\sigma_F\prod_{x\in F}a_x$.
For $K_3\sqcup K_1$, pairing the isolated vertex with any triangle
vertex leaves an anticommuting pair. No partition then contributes.

The remaining coefficients follow from
\begin{equation}
  s^2=\sum_{u\neq0}a_u^4+2\sum_{u<v}a_u^2a_v^2.
  \label{eq:s-square}
\end{equation}
Every squared-pair term initially has coefficient $d(d-2)/4$ in
Eq.~\eqref{eq:e4-cw}; the diagonal terms of $c_w^2$ subtract a
further $d$ for each commuting pair. Collecting these terms and the
three graph cases gives the general expansion
\begin{equation}
\begin{aligned}
 e_4(A)={}&\frac{d(d-2)}8\sum_{u\neq0}a_u^4
 +\frac{d(d-6)}4\!\sum_{\substack{u<v\\\omega(u,v)=0}}a_u^2a_v^2
 \\
 &+\frac{d(d-2)}4\!\sum_{\substack{u<v\\\omega(u,v)=1}}a_u^2a_v^2
 \\
 &-6d\sum_{F\in\FKempty}\tau_F\prod_{u\in F}a_u
 \\
 &-2d\sum_{F\in\FCfour}\sigma_F\prod_{u\in F}a_u.
\end{aligned}
 \label{eq:monomial-expansion}
\end{equation}
Only the last two sums involve four distinct labels. Their signs
depend on Pauli multiplication and on the coefficients $a_u$.

\subsection{Positive semidefinite decomposition}
\label{subsec:psd-decomposition}

We now rewrite the squared-pair terms from anticommuting pairs and the
$C_4$ terms in Eq.~\eqref{eq:monomial-expansion} as quadratic forms.
For each $w\neq0$, define the real column vector $z_w$, whose coordinates
are the products $a_ua_v$ indexed by the unordered pairs $\{u,v\}$ with
$u+v=w$ and $\omega(u,v)=1$, by
\begin{equation}
  z_w=\bigl(a_ua_v\bigr)_{\substack{u<v\\u+v=w,\ \omega(u,v)=1}}
      \in\mathbb R^{d^2/4}.
  \label{eq:zw}
\end{equation}
Since $v=u+w$ and $\omega(u,u+w)=\omega(u,w)$, the indexing condition
has $d^2/2$ solutions for $u$. The labels $u$ and $u+w$ determine the
same unordered pair, leaving $d^2/4$ coordinates. A coordinate is zero
whenever either of its labels lies outside the support.

Let $\alpha=\{u,v\}$ and $\beta=\{r,t\}$ be distinct indices of $z_w$.
Since both pairs have sum $w$, sharing a label would force the other
labels to coincide, so $\alpha$ and $\beta$ are disjoint. Furthermore,
because $u+v+r+t=0$ and the anticommutation graph of $\alpha\cup\beta$
contains two edges with no common endpoint, the classification in
Fig.~\ref{fig:anticommutation-graphs} shows that $\alpha\cup\beta$ is a
$C_4$ four-set, whose unique commuting partition
determines $\sigma_{\alpha\cup\beta}$. Define the real symmetric matrix
$H_w\in\mathbb R^{(d^2/4)\times(d^2/4)}$ by
\begin{equation}
  (H_w)_{\alpha,\beta}=
  \begin{cases}
    1,&\alpha=\beta,\\
    \sigma_{\alpha\cup\beta},&\alpha\neq\beta.
  \end{cases}
  \label{eq:Hw-entries}
\end{equation}

Consequently, with $I_{d^2/4}$ denoting the identity on
$\mathbb R^{d^2/4}$, these two terms become
\begin{equation}
\begin{aligned}
  &\sum_{w\neq0}z_w^{\mathsf T}
     \left(\frac{d^2}4I_{d^2/4}-\frac d2H_w\right)z_w
  \\
  &\quad=\frac{d(d-2)}4\!\sum_{\substack{u<v\\\omega(u,v)=1}}a_u^2a_v^2
      -2d\sum_{F\in\FCfour}\sigma_F\prod_{u\in F}a_u.
\end{aligned}
  \label{eq:coefficient-count}
\end{equation}
To check the coefficients in Eq.~\eqref{eq:coefficient-count}, we first
consider the diagonal terms. Since the diagonal entries of $H_w$ are
$1$, the coordinate square $(z_w)_{\{u,v\}}^2=a_u^2a_v^2$ has
coefficient
\begin{equation}
  \frac{d^2}4-\frac d2=\frac{d(d-2)}4.
  \label{eq:diagonal-coefficient}
\end{equation}
Every unordered anticommuting pair belongs to exactly one group,
namely $w=u+v$, so summing over $w$ produces the first sum on the
right-hand side without repetition. For the off-diagonal terms, let
$\alpha=\{u,v\}$ and $\beta=\{r,t\}$ be distinct indices. The two
symmetric matrix entries together contribute
\begin{equation}
  2\left(-\frac d2\sigma_{\alpha\cup\beta}\right)
    (a_ua_v)(a_ra_t)
  =-d\,\sigma_F\prod_{x\in F}a_x,
  \label{eq:offdiagonal-contribution}
\end{equation}
where $F=\alpha\cup\beta$. If the vertices of a $C_4$ four-set are
$u_1,u_2,u_3,u_4$ in cyclic order, its two partitions into
anticommuting pairs are
\begin{equation}
  \{\{u_1,u_2\},\{u_3,u_4\}\},\qquad
  \{\{u_1,u_4\},\{u_2,u_3\}\}.
  \label{eq:C4-anticommuting-partitions}
\end{equation}
The zero-sum condition makes the pair sums agree within each partition,
while the two partitions give distinct values of $w$.
Thus each $F\in\FCfour$ is counted twice, with total coefficient
$-2d\,\sigma_F$, which gives the second sum on the right-hand side.

The quadratic-form representation in Eq.~\eqref{eq:coefficient-count}
is useful because its coefficient matrices are positive semidefinite.
We establish this by using Clifford conjugation to reduce every nonzero
$w$ to a single canonical representative. This reduction preserves the
quartic quantity because unitary conjugation preserves the spectrum, and
hence
\begin{equation}
  e_4(UAU^\dagger)=e_4(A).
  \label{eq:e4-Clifford}
\end{equation}
A Clifford unitary also maps Pauli strings to signed Pauli strings. For a
Clifford $U$, write
\begin{equation}
  UP_uU^\dagger=\eta_U(u)P_{M_Uu},\qquad\eta_U(u)\in\{\pm1\}.
  \label{eq:Clifford-signs}
\end{equation}
The map $M_U$ preserves sums and commutation, so it sends the indexing
pairs of $z_w$ bijectively to those of $z_{M_Uw}$. For a commuting pair,
conjugating $P_uP_v=\epsilon(u,v)P_{u+v}$ gives
\begin{equation}
  \epsilon(M_Uu,M_Uv)
  =\eta_U(u)\eta_U(v)\eta_U(u+v)\epsilon(u,v).
  \label{eq:epsilon-transform}
\end{equation}
In a commuting partition of a zero-sum four-set, the two pair-sums agree,
so the two factors attached to that sum cancel. Therefore
\begin{equation}
  \sigma_{M_UF}=\left(\prod_{u\in F}\eta_U(u)\right)\sigma_F.
  \label{eq:sigma-transform}
\end{equation}
Here $M_UF=\{M_Uu:u\in F\}$. Write
\begin{equation}
  A':=UAU^\dagger=\sum_{v\neq0}a'_vP_v,
\end{equation}
and define $z'_{M_Uw}$ from the coefficients $a'_v$ as in
Eq.~\eqref{eq:zw}. Then
$a'_{M_Uu}=\eta_U(u)a_u$. Let $S_U$ be the signed permutation that sends
the coordinate indexed by $\{u,v\}$ to the coordinate indexed by
$\{M_Uu,M_Uv\}$ with sign $\eta_U(u)\eta_U(v)$. The transformed vector
and coefficient matrix satisfy
\begin{equation}
  z'_{M_Uw}=S_Uz_w,
  \qquad
  H_{M_Uw}=S_UH_wS_U^{\mathsf T}.
\end{equation}
Thus $H_w$ and $H_{M_Uw}$ are orthogonally similar.

Since the Clifford action is transitive on the nonzero Pauli labels, the
preceding orthogonal similarity reduces the calculation to a single
representative, which we take to satisfy
$P_w=X\otimes I^{\otimes(n-1)}$. With this choice, the anticommuting pairs
with sum $w$ are
\begin{equation}
  \{Z\otimes P_r,\;Y\otimes P_r\},\qquad r\in V_{n-1}.
  \label{eq:canonical-pairs}
\end{equation}
These $|V_{n-1}|=4^{n-1}=d^2/4$ pairs are indexed by $r$. For distinct
$r,s$, the unique commuting partition matches the two $Z$ factors and
the two $Y$ factors when $P_r$ and $P_s$ commute, giving sign $+1$.
When $P_r$ and $P_s$ anticommute, it matches each $Z$ factor with a
$Y$ factor, giving sign $-1$. Including the diagonal entries, we obtain
\begin{equation}
  (H_w)_{r,s}=(-1)^{\omega_{n-1}(r,s)},
  \label{eq:canonical-H}
\end{equation}
where $\omega_{n-1}$ is the symplectic form on the remaining qubits.
To compute $H_w^2$, observe that when $r=t$, the exponent
$\omega_{n-1}(r+t,s)$ vanishes for every $s$. When $r\neq t$,
nondegeneracy makes it a nonzero linear functional of $s$, so the
positive and negative terms cancel. Therefore
\begin{equation}
  (H_w^2)_{r,t}
    =\sum_{s\in V_{n-1}}(-1)^{\omega_{n-1}(r+t,s)}
    =\frac{d^2}4\delta_{r,t}.
  \label{eq:canonical-H-square}
\end{equation}
Orthogonal similarity carries this identity back to every nonzero $w$,
and hence
\begin{equation}
  \frac{d^2}4I_{d^2/4}-\frac d2H_w
    =\frac12\left(\frac d2I_{d^2/4}-H_w\right)^2\succeq0.
  \label{eq:coefficient-square}
\end{equation}
The right-hand side is positive semidefinite because
$(d/2)I_{d^2/4}-H_w$ is real symmetric.
Substituting Eqs.~\eqref{eq:coefficient-count} and
\eqref{eq:coefficient-square} into Eq.~\eqref{eq:monomial-expansion}
gives the exact quartic identity
\begin{equation}
\begin{aligned}
  &e_4(A)+6d\sum_{F\in\FKempty}\tau_F\prod_{u\in F}a_u
  \\
  &\quad=\frac{d(d-2)}8\sum_{u\neq0}a_u^4
  \\
  &\qquad+\frac{d(d-6)}4\!\sum_{\substack{u<v\\\omega(u,v)=0}}a_u^2a_v^2
  \\
  &\qquad+\frac12\sum_{w\neq0}
     \left\|\left(\frac d2I_{d^2/4}-H_w\right)z_w\right\|_2^2.
\end{aligned}
  \label{eq:master}
\end{equation}
The identity incorporates the noncommuting four-label contributions into
nonnegative squared norms while leaving the commuting zero-sum four-set
terms explicit. For every $n\geq1$, it holds for any traceless Hermitian
Pauli combination. For $d\geq8$, the coefficients of the remaining fourth
powers and squared-pair terms are both positive.

\section{Consequences for pure-state tomography}
\label{sec:tomography}

We now translate the quartic identity into conditions on the omitted
Pauli support. Excluding commuting zero-sum four-sets guarantees universal
UDA; for three qubits the condition is also necessary, whereas at larger
qubit numbers eigenvalue multiplicities prevent the same characterization.

\subsection{Four-set-free supports}
\label{subsec:four-set-free-certificate}

\begin{theorem}[Quartic certificate for pure-state tomography]
\label{thm:master}
Let $n\geq3$, $d=2^n$, and suppose that the omitted support $T$
contains no commuting zero-sum four-set. Then every nonzero operator
$A=\sum_{u\in T}a_uP_u$ with real coefficients $a_u$ satisfies
\begin{equation}
  e_4(A)\geq\frac{d(d-2)}8\sum_{u\in T}a_u^4>0.
  \label{eq:quartic-lower-bound}
\end{equation}
Consequently, the retained measurements uniquely determine every pure
state among all density operators (UDA), and hence among pure states
(UDP).
\end{theorem}

\begin{proof}
The lower bound follows from the quartic identity in
Eq.~\eqref{eq:master}. Let $A\in \mathcal V_T$ be nonzero, and set $a_u=0$ for $u\notin T$. Since $T$ is four-set-free, each commuting
zero-sum four-set contains at least one label outside $T$, so every
monomial in the corresponding sum in Eq.~\eqref{eq:master} vanishes.
For $n\geq3$ we have $d\geq8$, so the coefficient $d(d-6)/4$ of the
commuting squared-pair terms is positive, and the squared norms are
nonnegative. Dropping these contributions yields the lower bound in
Eq.~\eqref{eq:quartic-lower-bound}. Because the Pauli strings are linearly independent and $A\neq0$, at least one coefficient $a_u$ is
nonzero. The fourth-power sum, and hence $e_4(A)$, is therefore
strictly positive.

The tomography conclusion then follows from Eq.~\eqref{eq:e4-inertia}:
every nonzero $A\in \mathcal V_T$ has at least two positive and two negative
eigenvalues. These are exactly the inequalities required by the
UDA kernel criterion in Eq.~\eqref{eq:uda-kernel}. Physically, as
explained in Sec.~\ref{sec:measurement-kernel},
the difference between a pure state and any distinct density operator
has exactly one positive
eigenvalue and therefore cannot belong to $\mathcal V_T$.
The retained data thus identify every pure state among all
density operators, which
proves universal UDA and hence universal UDP.
\end{proof}

Theorem~\ref{thm:master} thus reduces the certification of universal
UDA to a finite combinatorial test on the omitted observables. The
pair-grouping test in Eq.~\eqref{eq:bucket-test} decides
four-set-freeness using only binary addition and Pauli commutation.
For $n\geq3$, a support that passes this test certifies universal UDA without any optimization over states or operator coefficients.

\subsection{Three qubits and the 945 minimal failing sets}
\label{subsec:three-qubit-failing-sets}

\begin{theorem}[Three-qubit Pauli tomography]
\label{thm:three-qubit}
For a three-qubit Pauli measurement scheme, universal UDP and universal
UDA are equivalent. Both hold if and only if the omitted support $T$
contains no commuting zero-sum four-set.
\end{theorem}

\begin{proof}
Sufficiency follows from
Theorem~\ref{thm:master}. For necessity, suppose that $T$ contains
a commuting zero-sum four-set $F$. At three qubits, the operator in
Eq.~\eqref{eq:fourset-factor} belongs to
$\mathcal V_F\subseteq\mathcal V_T$ and has
rank two by Eq.~\eqref{eq:fourset-rank}. Being Hermitian and traceless,
it has two nonzero eigenvalues $\lambda$ and $-\lambda$, with
$\lambda>0$, and can be written as
\begin{equation}
  A=\lambda\bigl(
    |\psi_+\rangle\langle\psi_+|
    -|\psi_-\rangle\langle\psi_-|
  \bigr),
  \label{eq:three-qubit-witness-states}
\end{equation}
where $|\psi_+\rangle$ and $|\psi_-\rangle$ are orthonormal.
Since $A$ lies in the kernel, $\Tr(P_vA)=0$ for every retained
Pauli observable $P_v$. Substituting
Eq.~\eqref{eq:three-qubit-witness-states} shows that the two pure
states have identical expectations of every retained observable.
Since they are distinct, the scheme fails UDP and
hence also UDA, which establishes necessity.
\end{proof}

Theorem~\ref{thm:three-qubit} provides an analytic characterization of the failing sets
studied in Refs.~\cite{MaEtAl2016,Jackson2013}. In the canonical
example in Eq.~\eqref{eq:standard-failing}, omitting $IIZ$, $IZI$, $ZII$, and
$ZZZ$ hides the population difference between $|000\rangle$ and
$|111\rangle$, as well as the mixing weight of any classical
mixture of these two states. Clifford transformations
generate the other
commuting zero-sum four-sets and map the indistinguishable
states by the same unitary. Since every three-qubit failing support contains
at least one such four-set, these four-sets are precisely the
inclusion-minimal failing supports.

The number of these minimal failing supports, $945$, follows directly from the binary description. We count ordered triples of independent, pairwise-commuting labels
$u,v,w\in V_3=\F_2^6$. There are $2^6-1=63$ choices for $u$.
The condition $\omega(u,v)=0$ leaves $2^5$ labels, and excluding $0$ and $u$ leaves $2^5-2=30$ choices for $v$. By the nondegeneracy of
$\omega$, commuting with both $u$ and $v$ imposes two
independent binary constraints on $w$, which leaves $2^4$
labels; excluding the four elements of
$\Span_{\F_2}\{u,v\}$ leaves $2^4-4=12$ choices. Each triple
determines the four-set
\begin{equation}
  F=\{u,v,w,u+v+w\}.
  \label{eq:four-set-from-triple}
\end{equation}
Conversely, any three elements of such a four-set are independent,
since
a dependence among them would force the fourth label to vanish.
Each unordered four-set is therefore counted $4\cdot3\cdot2=24$ times,
once for each ordered choice of three of its elements, and hence
\begin{equation}
  N_{\mathrm{four\text{-}sets}}
    =\frac{63\cdot30\cdot12}{24}=945.
  \label{eq:945-count}
\end{equation}
This count refers to distinct supports, all of which lie in the Clifford orbit
of the canonical example, up to individual Pauli signs.

The factorization $945=135\cdot7$ relates this count to common
eigenstates of commuting Pauli observables. Each four-set spans a unique three-dimensional commuting
label space $L$. Since such a space has $7\cdot6\cdot4$ ordered bases, the number of these spaces is
\begin{equation}
  N_L=\frac{63\cdot30\cdot12}{7\cdot6\cdot4}=135.
  \label{eq:135-commuting-spaces}
\end{equation}
Each $L$ labels a maximal commuting Pauli family together with its common
eigenbasis of eight states. The seven nonzero binary linear
functionals $\ell:L\to\F_2$ specify the seven four-sets contained in
$L$,
\begin{equation}
  F_\ell=\{u\in L:\ell(u)=1\}.
  \label{eq:four-set-functional}
\end{equation}
For two distinct states in this common eigenbasis, the ratios of
their Pauli eigenvalues form a nontrivial sign pattern
$(-1)^{\ell(u)}$. Every Pauli string outside $L$ anticommutes with some member
of the maximal commuting family and therefore has zero expectation in
both states. The expectation values of the two states thus differ exactly on
$F_\ell$, and the difference of their projectors has this Pauli support.
Each of the seven supports arises from four unordered pairs of
the eight states. The number $945$ therefore counts patterns of
omitted observables that hide these population differences, rather
than the state pairs themselves.

\subsection{Beyond three qubits}
\label{subsec:beyond-three-qubits}

Theorem~\ref{thm:master} remains valid for $n\geq4$, but the
necessity statement of Theorem~\ref{thm:three-qubit} does not extend.
The four-set witness in Eq.~\eqref{eq:fourset-factor} has
$d/8$ eigenvalues of each sign, so for $n\geq4$ its
normalized positive and negative parts are mixed states. Two
related four-qubit examples
illustrate how this multiplicity changes
the tomography problem.

Let $F_0=\{IIZ,IZI,ZII,ZZZ\}$ and $B=\sum_{P\in F_0}P$, so that
Eq.~\eqref{eq:standard-failing} gives
$B=4(|000\rangle\langle000|-|111\rangle\langle111|)$.
First, omit only
\begin{equation}
  T_4=\{IIZI,IZII,ZIII,ZZZI\},
  \label{eq:four-qubit-four-set}
\end{equation}
where we identify Pauli strings with their labels. This support is a commuting
zero-sum four-set, yet the complementary measurement scheme is UDA.
Indeed, every nonzero kernel operator has the form $C\otimes I_2$,
where $C$ is a nonzero traceless Hermitian combination of the
three-qubit strings in $F_0$. Since $C$ has eigenvalues of both
signs and tensoring with $I_2$ doubles the multiplicity of every
eigenvalue,
\begin{equation}
  n_\pm(C\otimes I_2)=2n_\pm(C)\geq2.
  \label{eq:spectator-inertia}
\end{equation}
Hence the UDA kernel criterion holds throughout this span.
In particular, $B\otimes I_2$ has rank four, with nonzero eigenvalues
$+4,+4,-4,-4$, and its normalized positive and negative parts are
$|000\rangle\langle000|\otimes I_2/2$ and
$|111\rangle\langle111|\otimes I_2/2$.

Next, additionally omit the four strings obtained by replacing the
final
$I$ in Eq.~\eqref{eq:four-qubit-four-set} with $Z$:
\begin{equation}
  T_8=\{P\otimes I,\;P\otimes Z:P\in F_0\}.
  \label{eq:four-qubit-eight-set}
\end{equation}
The enlarged kernel now contains
\begin{equation}
  B\otimes(I+Z)
  =8\bigl(|0000\rangle\langle0000|-|1110\rangle\langle1110|\bigr).
  \label{eq:four-qubit-rank-two}
\end{equation}
Consequently, the two pure states $|0000\rangle$ and $|1110\rangle$ have identical retained expectation values, so
both UDP and UDA fail. The additional omitted observables allow
the factor $I+Z=2|0\rangle\langle0|$ to select a single state of the
fourth qubit, thereby removing the multiplicity that protected UDA in
Eq.~\eqref{eq:spectator-inertia}.

Both omitted supports contain a commuting zero-sum four-set,
yet the
first yields a UDA scheme while the second fails even UDP.
Four-set-freeness therefore remains sufficient for $n\geq4$ but is no longer necessary.
These examples do not separate UDP from UDA, since the two properties
agree in each case. They show why deciding uniqueness for a
support containing a four-set requires information about the other
omitted observables.

Numerical evidence nevertheless suggests that UDP and UDA may
remain equivalent more generally for Pauli measurement schemes.
The variational searches reported in Ref.~\cite{ZhangZhuZeng2024}
found agreement between the two properties for four- and five-qubit systems
and identified schemes with $106$ and $393$ Pauli observables,
respectively, including the identity, that passed both numerical tests.
These observations support a broader equivalence but do not establish
it for arbitrary Pauli supports.

\section{Commuting Pauli kernels for an arbitrary number of qubits}
\label{sec:general}

The four-qubit examples in Sec.~\ref{subsec:beyond-three-qubits} show
that eigenvalue multiplicities can preserve UDA even when a commuting
four-set is omitted. For a general commuting support, Clifford
diagonalization reduces the kernel to basis-state populations. This turns
every pure--mixed ambiguity into a pure--pure one and yields both
UDP--UDA equivalence and exact spanning and rank criteria.

\subsection{Clifford reduction and pure-state uniqueness}
\label{subsec:commuting-diagonal-criterion}

Choose independent generators for the binary span of $T$. Since their
pairwise symplectic products vanish, they can be completed to a symplectic
basis. Consequently, a Clifford unitary $U$ maps every Pauli string in
$T$ to a signed product of $I$ and $Z$ factors
\cite{Gottesman1997,DehaeneDeMoor2003}. Absorbing these signs into
the real coefficients gives
\begin{equation}
  \begin{gathered}
    U \mathcal V_T U^\dagger=\Span_{\R}\{Z^u:u\in T_Z\},\\
    T_Z\subseteq\F_2^n\setminus\{0\},
    \qquad Z^u=\bigotimes_{j=1}^n Z^{u_j}.
  \end{gathered}
  \label{eq:commuting-diagonal-form}
\end{equation}
Here $T_Z$ records the transformed labels by their $n$ $Z$-bits alone. In the computational basis,
\begin{equation}
  Z^u|x\rangle=(-1)^{u\cdot x}|x\rangle,
  \qquad x\in\F_2^n,
  \label{eq:diagonal-pauli-outcomes}
\end{equation}
so each basis state gives a deterministic outcome for every diagonal
Pauli observable. The retained diagonal labels are
\begin{equation}
  S_Z=\F_2^n\setminus T_Z,
  \label{eq:measured-diagonal-labels}
\end{equation}
which include $u=0$, the label of the identity. Every nondiagonal Pauli string
is also retained, because all omitted strings are diagonal in these
coordinates. Since states and observables are
transformed by the same $U$, expectation values are preserved, as are kernel ranks and
inertia.

\begin{theorem}[Commuting Pauli kernels]
\label{thm:commuting}
For any pairwise-commuting omitted Pauli support $T$, the kernel
$\mathcal V_T$ contains an operator with exactly one positive or exactly one
negative eigenvalue if and only if it contains a nonzero rank-two
operator. Consequently, universal UDP and universal UDA are
equivalent, and both hold exactly when the retained diagonal labels
after Clifford conjugation satisfy
\begin{equation}
  \Span_{\F_2}S_Z=\F_2^n.
  \label{eq:diagonal-spanning-criterion}
\end{equation}
\end{theorem}

\begin{proof}
We work in the diagonal coordinates of
Eq.~\eqref{eq:commuting-diagonal-form}. Suppose that $\mathcal V_T$ contains a nonzero
operator $A$ with exactly one positive or exactly one negative eigenvalue.
Replacing $A$ by $-A$ if necessary, we may assume that it has exactly one
positive eigenvalue. Since $A$ is diagonal and traceless, it takes
the form
\begin{equation}
  \begin{gathered}
    A=\lambda\left(|x_0\rangle\langle x_0|
      -\sum_{x\neq x_0}p_x|x\rangle\langle x|\right),\\
    \lambda>0,\qquad p_x\geq0,\qquad
    \sum_{x\neq x_0}p_x=1.
  \end{gathered}
  \label{eq:diagonal-one-sided}
\end{equation}
This operator represents an ambiguity between a single computational
basis state and a mixture of other basis states. For every
$u\in S_Z$, Pauli orthogonality gives $\Tr(AZ^u)=0$, and hence
\begin{equation}
  (-1)^{u\cdot x_0}
    =\sum_{x\neq x_0}p_x(-1)^{u\cdot x}.
  \label{eq:deterministic-parity-average}
\end{equation}
The left-hand side equals $+1$ or $-1$, whereas the right-hand side is a weighted average of values in $\{+1,-1\}$. Such an average
attains an extremal value only if every term
with positive weight has
that value. Therefore every $x$
with $p_x>0$ satisfies
\begin{equation}
  (-1)^{u\cdot x}=(-1)^{u\cdot x_0},
  \qquad \forall\,u\in S_Z.
  \label{eq:matching-parities}
\end{equation}

Now choose any $y\neq x_0$ with $p_y>0$ and set
$B=|x_0\rangle\langle x_0|-|y\rangle\langle y|$.
To verify that this rank-two operator lies in the omitted span, we
expand it in the diagonal Pauli basis:
\begin{equation}
  B=\frac1{2^n}\sum_{u\in\F_2^n}
    \left[(-1)^{u\cdot x_0}-(-1)^{u\cdot y}\right]Z^u.
  \label{eq:diagonal-rank-two-expansion}
\end{equation}
By Eq.~\eqref{eq:matching-parities}, every coefficient with
$u\in S_Z$ vanishes. Hence $B\in\Span_{\R}\{Z^u:u\in T_Z\}$,
and conjugating back gives $U^\dagger B U\in\mathcal V_T$. The only
nonzero eigenvalues of this operator are $+1$ and $-1$, and the two indistinguishable
pure states in the original coordinates are $U^\dagger|x_0\rangle$
and $U^\dagger|y\rangle$. Conversely, every nonzero traceless
Hermitian rank-two operator has exactly one eigenvalue of each sign. Together with the kernel criteria in Eqs.~\eqref{eq:udp-kernel} and \eqref{eq:uda-kernel}, this establishes UDP--UDA equivalence.

The same argument yields the exact spanning criterion. If $S_Z$
spans $\F_2^n$, Eq.~\eqref{eq:matching-parities} forces
$u\cdot(y+x_0)=0$ for a spanning set of labels $u$, so $y=x_0$.
This contradicts the
choice of $y$; hence no one-sided inertia
obstruction exists, and the scheme is UDA. Conversely, suppose that $S_Z$
does not span $\F_2^n$,
and choose
\begin{equation}
  0\neq v\in\bigl(\Span_{\F_2}S_Z\bigr)^\perp,
  \label{eq:undetected-binary-direction}
\end{equation}
where the orthogonal complement is taken with respect to the ordinary binary dot product.
The distinct basis states $|0^n\rangle$ and $|v\rangle$ then give the
same outcomes for all retained diagonal Pauli observables. Setting $x_0=0$ and $y=v$
in Eq.~\eqref{eq:diagonal-rank-two-expansion} shows that their
projector difference is a rank-two kernel operator in these
coordinates. The scheme therefore fails UDP, which proves
Eq.~\eqref{eq:diagonal-spanning-criterion}.
\end{proof}

Physically, a mixture can reproduce the deterministic diagonal
Pauli outcomes of a basis state only if every state in
the mixture reproduces them individually. Any pure--mixed ambiguity
therefore
gives rise to a pure--pure ambiguity. The spanning condition specifies precisely when the retained parity measurements distinguish
all computational basis states.

\subsection{Eigenvalue multiplicities and coset ranks}
\label{subsec:commuting-ranks}

The diagonal picture also yields rank constraints and explains the eigenvalue
multiplicities observed in the four-qubit examples. An immediate consequence
of Theorem~\ref{thm:commuting} is that the
minimum nonzero rank in a commuting Pauli span cannot equal three.
Indeed, a traceless Hermitian rank-three
operator has
inertia $(1,2)$ or $(2,1)$, so the theorem supplies a rank-two operator
in the same span. More generally, suppose that the binary span of $T$ has dimension $k<n$. A Clifford unitary can then place its generators on the first $k$ qubits, and every kernel operator takes the form
\begin{equation}
  U A U^\dagger=D\otimes I_{2^{n-k}},
  \label{eq:commuting-spectator-form}
\end{equation}
where $D$ is diagonal and traceless, with diagonal entries constrained by the omitted support. Each eigenvalue of $D$ thus
occurs $2^{n-k}$ times in $A$. For nonzero $A$, both signs occur, each with multiplicity at least two, so the complementary measurement
scheme
for every such support is UDA. When $k=n$, there is no
automatic
multiplicity, and uniqueness is decided by
Eq.~\eqref{eq:diagonal-spanning-criterion}.

A related normal form determines all possible ranks within the span of a Pauli-group
coset. Let $\mathcal G$ be the abelian subgroup of the Pauli group
generated by $r$ commuting Hermitian Pauli operators with
independent binary labels. Every element of $\mathcal G$ is Hermitian,
$|\mathcal G|=2^r$, and $-I\notin\mathcal G$. Let $q$ be a Hermitian
Pauli operator that commutes with every element of $\mathcal G$ and whose label
lies outside the binary span of the generators.
The
resulting $r+1$ labels are independent and commuting, so $r+1\leq n$,
and a Clifford unitary maps them, up to signs, to $Z_1,\ldots,Z_r,Z_{r+1}$.
Every operator in the real span of the coset
$q\mathcal G=\{qg:g\in\mathcal G\}$
therefore takes the form
\begin{equation}
  U A U^\dagger=D\otimes Z\otimes I_{2^{n-r-1}},
  \label{eq:coset-diagonal-form}
\end{equation}
where $D$ is now an arbitrary real diagonal $2^r\times2^r$ matrix,
since the products of $Z_1,\ldots,Z_r$ form a basis of the
diagonal operators on those $r$ qubits.

Each nonzero diagonal entry $b$ of $D$ contributes the eigenvalues
$b$ and $-b$ to the spectrum, each repeated $2^{n-r-1}$ times. If exactly $j$
diagonal entries are nonzero, then
\begin{equation}
  \begin{gathered}
    \rank A=j\,2^{n-r}=j\,\frac d{|\mathcal G|},\\
    n_+(A)=n_-(A)=j\,2^{n-r-1},
    \qquad 1\leq j\leq|\mathcal G|.
  \end{gathered}
  \label{eq:coset-ranks}
\end{equation}
Since every value of $j$ in this range is attainable, the minimum
nonzero rank is $d/|\mathcal G|$, and every element has balanced positive
and negative inertia.

The four-operator support $\{q,qg_1,qg_2,qg_1g_2\}$,
where $q,g_1,g_2$ are independent commuting generators, corresponds to
$r=2$. Up to
Pauli signs, this is the support of the operator in Eq.~\eqref{eq:fourset-factor}.
Equation~\eqref{eq:coset-ranks} recovers the rank in
Eq.~\eqref{eq:fourset-rank} and, moreover, controls the entire four-operator
span: every nonzero element has at least $2^{n-3}=d/8$
eigenvalues of each sign. For $n\geq4$, the complementary
measurement scheme is therefore UDA. This result extends the explicit
four-qubit calculation in Eq.~\eqref{eq:spectator-inertia} to
arbitrary commuting zero-sum four-sets and all $n\geq4$.

\section{Discussion and outlook}
\label{sec:discussion}

The quartic certificate connects the structure of the omitted Pauli
observables to the identification of every pure state, even against
mixed-state alternatives. Four-set-freeness is sufficient for
universal UDA whenever $n\geq3$, and it yields an exact criterion for both UDP and UDA at
three qubits. The dimension dependence has a physical origin: the
pure-state population difference hidden by a commuting four-set at
three qubits becomes a difference between populations of degenerate
subspaces at larger $n$. Omitting such a four-set alone then preserves
UDA, so its presence no longer suffices to decide uniqueness.

The central open question is whether UDP implies UDA for every
Pauli measurement scheme. For $n\geq4$, any counterexample must
have a noncommuting omitted support that contains a commuting zero-sum
four-set, since Theorem~\ref{thm:master} excludes four-set-free
supports
and Theorem~\ref{thm:commuting} excludes commuting ones.
Physically, the question is whether the retained Pauli data can
exclude every competing pure state while still admitting a mixed
state. Whether the
ambiguity associated with a four-set can take this
form depends on the other omitted observables.

A natural next step is to develop stronger structural criteria for these
remaining supports that incorporate eigenvalue multiplicities and the
effect of additional noncommuting observables. Such criteria would
extend the finite four-set test into a broader tool for selecting and certifying
measurement schemes. A complementary question is whether the
quartic lower bound in Eq.~\eqref{eq:quartic-lower-bound} can yield
quantitative bounds on
reconstruction errors when the measured expectation values are subject
to uncertainty. Establishing such a connection would relate the
exact uniqueness guarantees to stability under measurement noise.

\bibliography{refs}

\end{document}